\documentclass[aps,twocolumn,showpacs,%
amsmath,%
reprint%
,preprintnumbers,%draft%
] {revtex4-1}

\usepackage{amsmath}
\usepackage{graphicx}
\usepackage{fancyhdr}
\usepackage{hyperref}

\font\bb=msbm10 at 12pt
\newcommand{\mathbb}[1]{\hbox{\bb #1}} %real numbers,etc.
\def\softt{{\leavevmode\setbox1=\hbox{t}%
\hbox to \wd1{t\kern-0.6ex{\char039}\hss}}}%cstocs
\newtheorem{claim}{Claim}[section]
\newtheorem{theorem}[claim]{Theorem}
\newtheorem{proposition}[claim]{Proposition}
\newtheorem{lemma}[claim]{Lemma}
\newtheorem{remark}[claim]{Remark}
\newtheorem{remarks}[claim]{Remarks}
\newtheorem{example}[claim]{Example}

\newenvironment{proof}[1][Proof]{\noindent\textbf{#1.} }{\ \rule{0.5em}{0.5em}}

\begin{document}

\title[Resonances on hedgehog manifolds]{Resonances on hedgehog manifolds}

\author{Pavel Exner$^{a,\,}$$^{b,\,}$$^*$, Ji\v r\'i Lipovsk\'y$^{c}$\\ \vspace{6mm}}

\affiliation{
{}$^{a}$Doppler Institute for Mathematical Physics and Applied Mathematics, \hbox{Czech Technical University, B\v{r}ehov\'{a}~7, 11519 Prague, Czechia}\\
{}$^{b}$\hbox{Department of Theoretical Physics, Nuclear Physics Institute AS CR, 25068 \v{R}e\v{z} near Prague, Czechia}\\
\hbox{{}$^{c}$Department of Physics, Faculty of Science, University of Hradec Kr\'{a}lov\'{e},}\\ \hbox{Rokitansk\'{e}ho 62, 50003 Hradec Kr\'{a}lov\'{e}, Czechia}\\
{}$^{*}$corresponding author:  \texttt{exner@ujf.cas.cz}
%e-mail:, \texttt{jiri.lipovsky@uhk.cz},\\
}

\date{\today}

\begin{abstract}
We discuss resonances for a nonrelativistic and spinless quantum particle confined to a two- or three-dimensional Riemannian manifold to which a finite number of semiinfinite leads is attached. Resolvent and scattering resonances are shown to coincide in this situation. Next we consider the resonances together with embedded eigenvalues and ask about the high-energy asymptotics of such a family. For the case when all the halflines are attached at a single point we prove that all resonances are in the momentum plane confined to a strip parallel to the real axis, in contrast to the analogous asymptotics in some metric quantum graphs; we illustrate it on several simple examples. On the other hand, the resonance behaviour can be influenced by a magnetic field. We provide an example of such a `hedgehog' manifold at which a suitable Aharonov-Bohm flux leads to absence of any true resonance, i.e. that corresponding to a pole outside the real axis.
\end{abstract}

\maketitle
\thispagestyle{fancy}

\begin{center}

\emph{With congratulations to Professor Miloslav Havl\'{\i}\v{c}ek on the occasion of his 75th birthday}

\end{center} 

%%%%%%%%%%%%%%%%%%%%%%%%%%%%%%%%%%%%%%%%%%%%%%%%%%%%%%%%%%
%%  INTRODUCTION                                        %%
%%%%%%%%%%%%%%%%%%%%%%%%%%%%%%%%%%%%%%%%%%%%%%%%%%%%%%%%%%
\section{Introduction}%\addtocounter{section}{1}

Study of quantum systems the configuration space of which is geometrically and topologically nontrivial proved to be a fruitful subject both theoretically and practically. A lot of attention has been paid to quantum graphs -- a survey and a guide to further reading can be found in \cite{BK, EKKST}. Together with that other systems have been studied which one can regard as generalization of quantum graphs where the `edges' may have different dimensions; using the theory of self-adjoint extensions one can construct operator classes which serve as Hamiltonians of such models \cite{ES1}.

One sometimes uses a pictorial term `hedgehog manifold' for a geometrical construct consisting of Riemannian manifolds of dimension two or three together with line segments attached to them. In this paper we consider the simplest situation when we have a single connected manifold to which a finite number of semiinfinite leads are attached --- one is especially interested in transport in such a system. Particular models of this type have been studied, e.g., in \cite{BEG, BG, ES2, ETV, Ki}. Again for the sake of simplicity we limit ourselves mostly to the situation when there are no external fields; the Hamiltonian will act as the negative second derivative on the halflines representing the leads and as Laplace-Beltrami operator on the manifold.

We have said that quantum motion on hedgehog manifolds can be regarded as a generalization of quantum graphs. It is therefore useful to compare similarities and differences of the two cases, and we will recall at appropriate places of the text how the claims look like when the Riemannian manifold is replaced by a compact metric graph.

The first question one has to pose when resonances are discussed is what is meant by this term. The two prominent instances are \emph{resolvent resonances} identified with poles of the analytically continued resolvent of the Hamiltonian and \emph{scattering resonances} where we look instead into the analytical structure of the on-shell scattering operator. While the two often coincide, in general in may not be so; recall that the former are the property of a single operator while the latter refer to the pair of the full and unperturbed Hamiltonians, and often also a third one, an identification operator, which one uses if the two Hamiltonians act on different Hilbert spaces \cite{RS}.

The first question we will thus address deals with the two resonance definitions for quantum motion on hedgehog manifolds. Using an exterior complex scaling we will show that in this case both notions coincide and one is thus allowed to speak about resonances without a further specification. The result is the same as for quantum graphs \cite{EL1, EL2} and, needles to say, in many other situations.

The next question to be addressed in this paper concerns the high energy behaviour of the resonances which is, for this purpose, useful to count together with the eigenvalues. Note that if a hedgehog manifold with a finite number of junctions is compact having finite line segments, its spectrum is purely discrete and an easy estimate yields the spectral behaviour at high energies. It follows the usual Weyl's law \cite{We}, and moreover, it is determined by the manifold component with the highest dimension, that is, in our case the Riemannian manifold \cite{Iv}.

If, on the other hand, the leads are semiinfinite, the essential spectrum covers the positive real axis. In contrast to the usual Schr\"odinger operator theory it often contains embedded eigenvalues; it happens typically when the Laplace-Beltrami operator which is the manifold part of the Hamiltonian has an eigenfunction with zeros at the hedgehog junctions. Since such eigenvalues are unstable --- a geometrical perturbation turns them generically into resonances --- it is natural to count them together with the `true' resonances; one then asks about the asymptotics of the number of such singularities enclosed in the circle of radius $R$ in the momentum plane.

This question is made intriguing by the recent observation \cite{DEL, DP} that in some quantum graphs the asymptotics may not be of Weyl type. The reason behind this effect is that symmetries, maybe not apparent ones, may effectively diminish the graph size making a part of it effectively belonging to a lead instead. The mechanism uses the fact that all the edges of a quantum graph are one-dimensional, and one may expect that such a thing would not happen on hedgehog manifolds where the particles is forced to `change dimension' at the junctions. We are going to give a partial confirmation of this conjecture by showing that, in contrast to the quantum-graph case, the resonances cannot be found at arbitrary distance from real axis in the momentum plane as long as the leads are attached at a single point of the manifold.

The third and the last question addressed here is again inspired by an observation about quantum graphs. It has been noted that a magnetic field can change the effective size of some quantum graphs with a non-Weyl asymptotics \cite{EL3}: if we follow the resonance poles as functions of the field we observe that at some field values they move to (imaginary) infinity and the resonances disappear. On hedgehog manifolds the situation is different, nevertheless a similar effect may again occur; we will present a simple example of such a system in which a suitable Aharonov-Bohm field removes all the `true' resonances, i.e. those with pole position having nonzero imaginary part.

%%%%%%%%%%%%%%%%%%%%%%%%%%%%%%%%%%%%%%%%%%%%%%%%%%%%%%%%%%
%%  MODEL                                               %%
%%%%%%%%%%%%%%%%%%%%%%%%%%%%%%%%%%%%%%%%%%%%%%%%%%%%%%%%%%
\section{Description of the model}%\addtocounter{section}{1}

Let us first give a proper meaning to what we described above as quantum motion on a hedgehog manifold; doing so we generalize previously used definitions --- see, e.g. \cite{BEG, BG, ETV} --- by allowing more than a single semiinfinite lead be attached at a point of the manifold. Consider a compact and connected Riemannian manifold $\Omega \in \mathbb{R}^N,\: N =2,3$, endowed with metric $g_{rs}$. The manifold may or may not have a boundary, in the latter case we suppose that $\partial \Omega$ is smooth.

We denote by $\Gamma$ the geometric object consisting of $\Omega$ and a finite number $n_j$ of halflines attached at points $x_j,\: j=1, \dots, n$ belonging to a finite subset $\{x_j\}$ of the interior of $\Omega$ --- see figure \ref{fig1} --- we will employ the term \emph{hedgehog manifold}, or simply manifold if there is no danger of misunderstanding. By $M = \sum_j n_j$ we denote the total number of the halflines. The Hilbert space we are going to consider consists of direct sum of the `component' Hilbert spaces, in other words, its elements are square integrable functions on every component of $\Gamma$,
 % ------------- %
 \[
   \mathcal{H} = L^2 (\Omega, \sqrt{|g|} \,\mathrm{d}x) \,\oplus\, \bigoplus_{i=1}^M L^2 \big(\mathbb{R}_+^{(i)}\big)\,,
 \]
 % ------------- %
where $g$ stands for $\mathrm{det\,}(g_{rs})$ and $\mathrm{d}x$ for Lebesque measure on $\mathbb{R}^N$.

\begin{figure}
   \begin{center}
     \includegraphics{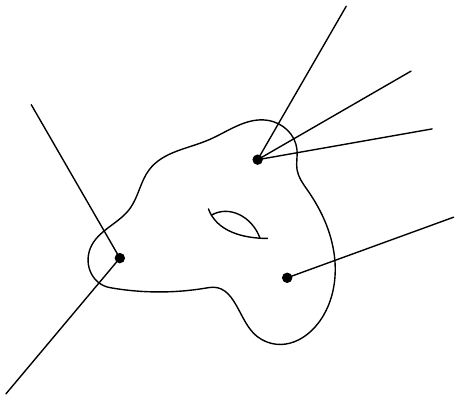}
     \caption{Example of a hedgehog manifold}\label{fig1}
   \end{center}
\end{figure}

Let $H_0$ be the closure of the Laplace-Beltrami operator $- g^{-1/2} \partial_r (g^{1/2} g^{rs} \partial_s)$ with the domain consisting of functions in $C_0^\infty(\Omega)$; if the boundary of $\Omega$ is nonempty  we require that they satisfy at it appropriate boundary conditions, either Neumann/Robin, $(\partial_n +\gamma) f|_{\partial \Omega} = 0$, or Dirichlet, $f|_{\partial \Omega} = 0$. (As mentioned above, we make this assumption for the sake of simplicity and most considerations below extend easily to Schr\"odinger type operators $- g^{-1/2} \partial_r (g^{1/2} g^{rs} \partial_s) +V(x)$ provided the potential $V$ is sufficiently regular.) The domain of $H_0$ coincides with $W^{2,2}(\Omega)$ which, in particular, means that $f(x)$ makes sense for $f\in D(H_0)$ and $x\in\Omega$. The restriction $H_0'$ of $H_0$ to the domain $\{f\in D(H_0):\: f(x_j) = 0,\: j=1,\dots,n\}$ is a symmetric operator with deficiency indices $(n,n)$, cf. \cite{BEG, BG}.
Furthermore, we denote by $H_i$ the negative Laplacian on $L^2(\mathbb{R}_+^{(i)})$ referring to the $i$-th halfline and by $H_i'$ its restriction to functions which vanish together with their first derivative at the halfline endpoint. Since each $H_i'$ has deficiency indices $(1,1)$, the direct sum $H' = H_0' \oplus H_1' \oplus \cdots \oplus H_M'$ is a symmetric operator with deficiency indices $(n+M,n+M)$.

The family of admissible Hamiltonians for quantum motion on the hedgehog manifold $\Gamma$ can be identified with the self-adjoint extensions of the operator $H'$. The procedure how to construct them  using the boundary-value theory was described in detail in \cite{BG}. It is a modification of the analogous result from the quantum graph theory \cite{Ha}, and in a broader context of a known general result \cite{GG}. All the extensions are described by the coupling conditions
 % ------------- %
 \begin{equation}
   (U -I) \Psi + i (U+I) \Psi' = 0\,, \label{eq-coupling}
 \end{equation}
 % ------------- %
where $U$ is an $(n+M)\times (n+M)$ unitary matrix, $I$ the corresponding unit matrix and $\Psi = (d_1(f), \dots ,d_n(f), f_1(0),\dots , f_n(0))^\mathrm{T}$, $\Psi' = (c_1(f), \dots , c_n(f), f_1'(0),\dots , f_n'(0))^\mathrm{T}$ are the columns of (generalized) boundary values. The first $n$ entries correspond to the manifold part being equal to the leading and next-to-leading terms of the asymptotics of $f(x)$ on $\Omega$ in the vicinity of $x_j$, while $f_i(0), f_i'(0)$ describe the limits of the wave function and its first derivative on $i$-th halfline, respectively. More precisely, according to Lemma~4 in \cite{BG}, for $f \in D(H_0^*)$ the asymptotic expansion near $x_j$ has the form
$f(x) = c_j(f) F_0 (x, x_j) + d_j(f) + \mathcal{O}(r(x,x_j))$, where
 % ------------- %
  \begin{equation}
   F_0 (x, x_j) = \left\{\begin{array}{lll}
-\frac{q_2(x,x_j)}{2 \pi}\, \ln{r(x,x_j)} & \quad \dots \quad & N = 2 \\ [.5em]
\frac{q_3(x,x_j)}{4 \pi}\, (r(x,x_j))^{-1} & \quad \dots \quad & N = 3\\
\end{array}\right.\label{eq-asym}
  \end{equation}
 % ------------- %
Here $q_2, q_3$ are continuous functions of $x$ with $q_i(x_j, x_j) = 1$ and $r(x,x_j)$ denotes the geodetic distance between $x$ and $x_j$. Function $F_0$ is the leading term, independent of energy, of the Green's function asymptotics near $x_j$, i.e.
 % ------------- %
 \[
   G(x,x_j;k) = F_0(x,x_j) + F_1(x,x_j;k) + R(x,x_j;k)
 \]
 % ------------- %
with the remainder term $R(x,x_j;k) =\mathcal{O}(r(x,x_j))$. The self-adjoint extension of $H'$ determined by the condition (\ref{eq-coupling}) will be denoted as $H_U$; we will drop the subscript if the choice of $U$ is either clear from the context or not important.

Not all self-adjoint extensions, however, make in general sense from the physics point of view. The reason is that for $n>1$ one finds among them such extensions which would allow the particle living on $\Gamma$ to hop from a junction to another one. We restrict our attention in what follows to \emph{local} couplings for which such a situation cannot occur. They are described by matrices which are block diagonal, so that such a $U$ does not connect disjoint junction points $x_j$. The coupling condition (\ref{eq-coupling}) is then a family of $n$ conditions, each referring to a particular $x_j$ and coupling the corresponding (sub)columns $\Psi_j$ and $\Psi_j'$ by means of the respective block $U_j$ of $U$.

Before proceeding further we will mention a useful trick, known from quantum-graph theory \cite{EL2}, which allows to study a compact scatterer with leads looking at its `core' alone replacing the leads by effective coupling at the points $x_j$ which is a non-selfadjoint, energy-dependent point interaction, namely
 % ------------- %
 \begin{equation}
   (\tilde U_j(k) - I) d_j(f) + i (\tilde U_j(k) + I) c_j(f) = 0\,,\label{eq-coupling2}
 \end{equation}
 \begin{multline*}
   \tilde U_j(k) = U_{1j} -(1-k) U_{2j} [(1-k) U_{4j} \\
   - (k+1) I]^{-1} U_{3j}\,,
 \end{multline*}
 % ------------- %
where $U_{1j}$ denotes the top-left entry of $U_j$, $U_{2j}$ the rest of the first row, $U_{3j}$ the rest of the first column and $U_{4j}$ is $n_j\times n_j$ part corresponding to the coupling between the halflines attached to the manifold. This can be easily checked using the standard argument ascribed, in different fields, to different people such as Schur, Feshbach, Grushin, etc.

%%%%%%%%%%%%%%%%%%%%%%%%%%%%%%%%%%%%%%%%%%%%%%%%%%%%%%%%%%
%%  RESONANCES                                          %%
%%%%%%%%%%%%%%%%%%%%%%%%%%%%%%%%%%%%%%%%%%%%%%%%%%%%%%%%%%
\section{Scattering and resolvent resonances}\label{sec-2}%\addtocounter{section}{1}

The model described in the previous section provides a natural framework to study scattering on the hedgehog manifold. The existence of scattering is easy to establish because any Hamiltonian of the considered class differs from one with decoupled leads by a finite-rank perturbation in the resolvent \cite{RS}. Finding the on-shell scattering matrix is computationally slightly more complicated but simple in principle. The solution of Schr\"odinger equation on the $j$-th external lead with energy $k^2$ can be expressed as a linear combination of the incoming and outgoing waves, $a_j(k) \mathrm{e}^{-ikx}+b_j(k) \mathrm{e}^{ikx}$. The scattering matrix then maps the vector of incoming wave amplitudes $a_j$ into the vector of outgoing wave amplitudes $b_j$.

We emphasize that our convention, which is natural in this context and analogous to the one used in quantum-graph theory \cite{KS1} differs from the one employed when scattering on the real line is treated \cite{TZ}, in that each lead is identified with the positive real halfline. For the case of two leads, in particular, it means that columns of the $2\times 2$ scattering matrix are interchanged and we have
 % ------------- %
 \begin{lemma}\label{l1}
 The on-shell scattering matrix satisfies
 % ------------- %
 \[
   S(k)^{-1} = S(-k) = S^*(\bar k)\,,
 \]
 % ------------- %
where star and bar denote the Hermitian and complex conjugation, respectively.
 % ------------- %
 \end{lemma}
 \begin{proof}
The claim follows directly from the definition of scattering matrix and from properties of Schr\"odinger equation and its external solutions. Since the potential is absent we infer that if $f(k_0,x)$ is a solution of the Schr\"odinger equation for a given $k$, so is $f(-k_0,x)$. This means that $S(k)$ can be regarded both as an operator mapping $\{a_j(k)\}$ to $\{b_j(k)\}$ or as a map from $\{b_j(-k)\}$ to $\{a_j(-k)\}$, i.e. as the inverse of $S(-k)$. In an similar way one can establish the second identity.
 \end{proof}
 % ------------- %

 % ------------- %
 \begin{remark}\label{rem1}
{\rm If $\Omega$ is replaced by a compact metric graph and the potential is again absent, the S-matrix can be written as $S(k) = - F(k)^{-1} \cdot F(-k)$ where the $M\times M$ matrix $F(k)$ is an analogue of Jost function.}
 \end{remark}
 % ------------- %

\noindent In particular, $S(k)$ is unitary for $k\in\mathbb{R}$, however, we will need it also for complex values of $k$. By a \emph{scattering resonance} we conventionally understand a pole of the on-shell scattering matrix in the complex plane, more precisely, the point at which some of its entries have a pole singularity.

A \emph{resolvent resonance}, on the other hand, is identified with a pole of the resolvent analytically continued from the upper complex halfplane to a region of the lower one. A convenient and efficient way of treating resolvent resonances is the method of exterior complex scaling based on the ideas of  Aguilar, Baslev, Combes, and Simon --- cf.~\cite{AC, BC, Si},  for a recent application to the case of quantum graphs see \cite{DP, EL1, EL2}). Resonances in this approach become eigenvalues of the non-selfadjoint operator $\mathrm{H}_\theta = U_\theta \mathrm{H} U_{\theta}^{-1}$ obtained by scaling the Hamiltonian outside a compact region with the scaling parameter taking a complex value $\mathrm{e}^{\theta}$; if $\mathrm{Im\,}\theta$ is large enough, the rotated essential spectrum reveals a part of the `unphysical' sheet with the poles being now true eigenvalues corresponding to square integrable eigenfunctions.

In the present case we identify, in analogy with the quantum-graph situation mentioned above, the exterior part of $\Gamma$ with the leads, and scale the wave function at each them using the transformation $(U_\theta f) (x) = \mathrm{e}^{\theta/2} f(\mathrm{e}^{\theta}x)$ which is, of course, unitary for real $\theta$, while for a complex $\theta$ it leads to the desired rotation of the essential spectrum. To use it we first state a useful auxiliary result.

 % ------------- %
 \begin{lemma}\label{lem-green}
Let $H|_\Omega$ be the restriction of an admissible Hamiltonian to $\Omega$ and suppose that $f(\cdot,k)$ satisfies $H|_\Omega f(x,k) = k^2 f(x,k)$ for $k^2 \not \in \sigma(H_0)$, then it can be written as a particular linear combination of Green's functions of $H_0$, namely
 % ------------- %
  \[
    f(x,k) =  \sum_{j=1}^{n} c_j G(x,x_j;k)\,.
  \]
 % ------------- %
 \end{lemma}
 % ------------- %
 \begin{proof}
The claim is a straightforward generalization of Lemma~2.2 in \cite{Ki} to the situation where \mbox{$n>2$} and more general couplings are imposed at the junctions. Suppose that $k$ is not an eigenvalue of $H_0$. The Green's functions with one argument fixed at different points $x_j$ are clearly linearly independent, hence $\mathrm{Dom\,}H'^* = W^{2,2}(\Omega) \oplus (\mathrm{Span\,}\{ G(x,x_j;k)\}_{j=1}^n)$, and without loss of generality one can write $f(x,k) =  \sum_{j=1}^n c_j G(x,x_j;k) + g(x)$ with $g\in W^{2,2}(\Omega)$. However, then we would have $H|_\Omega g(x) = H_0 g(x) = k^2 g(x)$ for $x\ne x_j$, and since $k^2 \not \in \sigma(H_0)$ and $g\in W^{2,2}(\Omega)$ by assumption, it follows that $g=0$.
 \end{proof}
 % ------------- %

 % ------------- %
 \begin{theorem}
In the described setting, the hedgehog system has a scattering resonance at $k_0$ with $\mathrm{Im\,}k_0<0$ and $k_0^2 \not \in \mathbb{R}$ \emph{iff} there is a resolvent resonance at $k_0$. Algebraic multiplicities of the resonances defined in both ways coincide.
 \end{theorem}
  % ------------- %
 \begin{proof}
Consider first the scattering resonances. The starting point is the generalized eigenfunction describing the scattering at the energy $k^2$ and its analytical continuation to the lower complex halfplane. From the previous lemma we know that for $k^2 \not \in \sigma(H_0)$ the restriction of the appropriate Schr\"odinger equation solution to the manifold is a linear combination of at most $n$ Green's functions; we denote the corresponding vector of coefficients by $\mathbf{c}$. The relation between these coefficients and amplitudes of the outgoing and incoming wave is given by (\ref{eq-coupling}). Using $\mathbf{a}$ as a shortcut for the vector of the amplitudes of the incoming waves, $(a_1(k), \dots, a_M(k))^\mathrm{T}$, and similarly $\mathbf{b}$ for the vector of the amplitudes of the outgoing waves one obtains in general system of equations
 % ------------- %
 \begin{equation} \label{scattres}
  A(k)  \mathbf{a} + B(k)  \mathbf{b} + C(k)\mathbf{c}  = 0\,,
 \end{equation}
 % ------------- %
in which $A$ and $B$ are $(n+M) \times M$ matrices and $C$ is $(n+M) \times n $ matrix the elements of which are exponentials and Green's functions, regularized if needed --- recall that $n$ is the number of internal parameters associated with the junctions and $M$ is the number of the leads. What is important that all the entries of the mentioned matrices allow for an analytical continuation which makes it possible to ask for solution of equations (\ref{scattres}) for $k=k_0$ from the open lower complex halfplane.

It is obvious that for $k_0^2 \not \in \mathbb{R}$ the columns of $C(k_0)$ have to be linearly independent; otherwise $k_0^2$ would be an eigenvalue of $H$ with an eigenfunction supported on the manifold $\Omega$ only. Hence there are $n$ linearly independent rows of $C(k_0)$ and after a rearrangement in equations (\ref{scattres}) one is able to express $\mathbf{c}$ from the first $n$ of them. Substituting then to the remaining equations one can rewrite them in the form
 % ------------- %
 \begin{equation}
   \tilde A(k_0)  \mathbf{a} + \tilde B(k_0)  \mathbf{b} = 0 \label{eq2}
 \end{equation}
 % ------------- %
with $\tilde A(k_0)$ and $\tilde B(k_0)$ being  $M\times M$ matrices the entries of which are rational functions of the entries of the previous ones. Suppose that $\mathrm{det\,}\tilde A(k_0) = 0$, then there exist a solution of the previous equation with $\mathbf{b} = 0$, and consequently, $k_0$ is an eigenvalue of $H$ since $\mathrm{Im\,}k_0<0$ and the corresponding eigenfunction belongs to $L^2$, however, this contradicts to the self-adjointness of $H$. Now it is sufficient to notice that the S-matrix analytically continued to the point $k_0$ equals $-\tilde B(k_0)^{-1} \tilde A(k_0)$ hence its pole singularities are solutions of the equation $\det \tilde B(k) = 0$.

Let us turn to resolvent resonances and consider the exterior complex scaling transformation $U_\theta$ with $\mathrm{Im\,}\theta>0$ large enough to reveal the sought pole on the second sheet of the energy surface. Choosing $\mathrm{arg\,}\theta > \mathrm{arg\,}k_0$ we find that the solution $a_j(k) \mathrm{e}^{-ikx}$ on the $j$-th lead, analytically continued to the point $k=k_0$, is after the transformation by $U_\theta$ exponentially increasing, while $b_j(k) \mathrm{e}^{ikx}$ becomes square integrable. This means that solving in $L^2$  the eigenvalue problem for the non-selfadjoint operator $\mathbf{H}_\theta$ obtained from $\mathbf{H}=H_U$ one has to find solutions of (\ref{eq2}) with \mbox{$\mathbf{a} = 0$}. This leads again to the condition  $\mathrm{det\,}\tilde B(k) = 0$ concluding thus the proof.
 \end{proof}
 % ------------- %

 % ------------- %
 \begin{remarks}{\rm
(a) It may happen, of course, that at some junctions the leads are disconnected from the manifold since the conditions (\ref{eq-coupling}) are locally separating and define a point interaction at those points, or that a junction coincides with a zero of an eigenfunction of $H_0$. In such situations it may happen that $H$ has eigenvalues, either positive, embedded into the continuous spectrum, or negative ones. In terms of the momentum variable $k$ these eigenvalues appear in pairs symmetric w.r.t. the origin.

\smallskip

\noindent (b) In the case of separating conditions (\ref{eq-coupling}) it may happen that the complex-scaled operator $\mathbf{H}_\theta$ has an eigenvalue in $k_0^2$  with eigenfunction supported outside $\Omega$, then $k_0$ is also a pole of $S(k)$; poles multiplicities of the resolvent and the scattering matrix may differ in this situation.

\smallskip

\noindent (c) In the quantum-graph analogue when $\Omega$ is replaced by a compact metric graph the decomposition of Lemma~\ref{lem-green} cannot be used since the deficiency indices of $H_0'$ may, in general, exceed $n$ and one can have extensions with the wave functions discontinuous at the junctions. The role of the internal parameters is instead played by the coefficients of two linearly independent solutions on each (internal) edge.
 }\end{remarks}
 % ------------- %
 
%%%%%%%%%%%%%%%%%%%%%%%%%%%%%%%%%%%%%%%%%%%%%%%%%%%%%%%%%%
%%  RESONANCES ASYMPTOTICS                              %%
%%%%%%%%%%%%%%%%%%%%%%%%%%%%%%%%%%%%%%%%%%%%%%%%%%%%%%%%%%
\section{Resonance asymptotics}%\addtocounter{section}{1}

The aim of this section is to say something about the asymptotic behaviour of the resolvent poles with respect to an increasing family of regions which cover in the limit the whole complex plane. Using Lemma~\ref{lem-green} let us write the manifold part $f(\cdot,k)$ of a function from the deficiency subspace of $H'$ as a linear combination of Green's functions of $H_\Omega$, acting as negative Laplace-Beltrami operator on $\Omega$. For $k^2 \not \in \sigma(H_\Omega)$  and $x$ in the vicinity of the point $x_i$ we then have
 % ------------- %
  \begin{multline*}
    f(x,k) = \sum_{j= 1}^n c_j G(x,x_j;k) = c_i F_0(x,x_i) \\
    + c_i F_1(x,x_i;k) + \sum_{i\ne j= 1}^n c_j G(x_i,x_j;k) + \mathcal{O}(r(x,x_i))
  \end{multline*}
 % ------------- %
which makes it easy to find the generalized boundary values $c_i(f)$ and $d_i(f)$ to be inserted into the coupling conditions (\ref{eq-coupling}), or effective conditions (\ref{eq-coupling2}).

We will employ the latter with matrix $\tilde U(k) = \mathrm{diag\,}(\tilde U_1(k),\dots,\tilde U_n(k))$ whose blocks correspond to junctions of $\Gamma$. We introduce
 % ------------- %
  \[
    Q_0 (k) = \left\{\begin{array}{lll}G(x_i,x_j;k) & \quad \dots & i \ne j \\[.3em] F_1(x_i,x_i;k) & \quad \dots & i = j
\end{array}\right.
  \]
 % ------------- %
which allows us to write $\textbf{d}(f) = Q_0(k) \textbf{c}$, and substituting into (\ref{eq-coupling2}) we can write the solvability of the system which determines the resonances as
 % ------------- %
  \begin{equation}
    \det\left[(\tilde U(k) - I) Q_0(k)+ i (\tilde U(k) + I)\right] = 0\,. \label{eq-rescon1}
  \end{equation}
 % ------------- %
We note that the matrices $\tilde U_j(k)$ entering this condition may be singular, however, this may happen for at most $M$ values of $k$, taking all the conditions together.

We also note that if the Hamiltonian $\mathbf{H}$ has an eigenvalue $k^2$ embedded in its continuous spectrum covering the interval $\mathbb{R}^+$, the corresponding $k>0$ also solves the equation (\ref{eq-rescon1}). Hence, as we have indicated in the introduction, from now on --- for purpose of this section --- we will include such embedded eigenvalues among resonances.

Having formulated the resonance condition we can ask how are its solutions distributed. To count zeros of a meromorphic function we employ the following auxiliary result.
 % ------------- %
 \begin{lemma}\label{lem-counting}
 Let $g$ be meromorphic function in $\mathbb{C}$ and suppose that it has no pole or zero on the circle $\mathcal{C}_R = \{z:\:|z| = R\}$. Then the difference between the number of the zeros and poles of $g$ in the disc of radius $R$ of which $\mathcal{C}_R$ is the perimeter is given by
 % ------------- %
 \[
   \int_{\mathcal{C}_R} \frac{g(z)'}{g(z)} \,\mathrm{d}z\,,
 \]
 % ------------- %
with prime denoting the derivative with respect to $z$, or equivalently, it is the difference between the number of jumps of the phase of $g(z)$ from $2\pi$ to $0$ along the circle $\mathcal{C}_R$ and the jumps from $0$ to $2\pi$.
 \end{lemma}
 % ------------- %
 \begin{proof}
Suppose that $g(\cdot)$ has at $z_0$ a zero of multiplicity $s$, i.e. $g(z) = h(z) (z-z_0)^s$ with neither $h(z)$ nor $h(z)'$ having a zero or a pole at $z_0$. Using
 % ------------- %
 \[
   \frac{g(z)'}{g(z)} =  \frac{h(z)'}{h(z)} + \frac{s}{z-z_0}
 \]
 % ------------- %
we find
 % ------------- %
 \begin{multline*}
   \mathrm{Res}_{z_0} \frac{g(z)'}{g(z)}\\
   = \frac{1}{(s-1)!}\lim_{z \to z_0}\frac{\mathrm{d}^{s-1}}{\mathrm{d}z^{s-1}}\left((z-z_0)^s \frac{g(z)'}{g(z)}\right) = s\,;
 \end{multline*}
 % ------------- %
a similar result differing only by the sign is obtained in the case of a pole of multi\-plicity~$s$. Consequently, the function $g(\cdot)'/g(\cdot)$ has pole with the residue $s$ if $g(\cdot)$ has zero of multiplicity $s$, and it has pole with the residue $-s$ at points where $g(\cdot)$ has pole of multiplicity~$s$. Furthermore, $g(\cdot)'$ does not have a pole at $z_0$ as long as $g(\cdot)$ does not which can be easily seen from the appropriate Laurent series. Using now the residue theorem we arrive at the desired integral expression. The claim about the number of phase jumps follows from the fact that $g'(z)/g(z) = (\mathrm{Ln}\,g(z))'$.
 \end{proof}
 % ------------- %

\subsection{Manifolds with the leads attached at a single point}

We shall consider the situation when $\Gamma$ has a single junction, i.e. there is a point $x_0\in\Omega$ at which all the $M$ halfline leads are attached. Then the matrix function $Q_0(k)$ is reduced to dimension one and it coincides with the regularized Green's function $F_1(x_0,x_0;k)$; for simplicity we drop in the rest of this subsection $x_0$ from the argument. The resonance condition (\ref{eq-rescon1}) then becomes $(\tilde u(k)-1) F_1(k) + i (\tilde u(k)+1) = 0$; we use the lower-case symbol to stress that $\tilde U(k)$ is just a number in this case.

The aim is now to establish the high-energy asymptotics. If we exclude the case of $\tilde u(k)=1$ when the leads are obviously decoupled from the manifold and the motion on $\Omega$ is described by the Hamiltonian $H_0$, we can without loss of generality rewrite the resonance condition as
 % ------------- %
 \[
   F(k) := F_1(k)+ i\, \frac{\tilde u(k) + 1}{\tilde u(k)- 1} = 0\,.
 \]
 % ------------- %
From (\ref{eq-coupling2}) we see that $\tilde u(\cdot) - I$ is a rational function. Consequently, it may add zeros or poles to those of $F_1(\cdot)$, however,
their number is finite and bounded by $M$ and $n=1$, respectively. The main thing is thus to find the behaviour of $F_1(k)$, in particular, its asymptotics for $k$ far enough from the real axis.

 % ------------- %
 \begin{lemma}\label{lem-asym}
The asymptotics of regularized Green's function is the following:
\begin{enumerate}
 % ------------- %
\item For $d=2$ we have  $F_1(k) = \frac{1}{2 \pi} \big(\ln(\pm ik) - \ln{2}-\gamma_\mathrm{E})+\mathcal{O}(|\mathrm{Im\,}k|^{-1} \big)$  if $\mp\mathrm{Im\,}k>0$,
 % ------------- %
\item For $d=3$ we have  $F_1(k) = \pm \frac{ik}{4\pi} +\mathcal{O}(|\mathrm{Im\,}k|^{-1} \big)$  if $\mp\mathrm{Im\,}k>0$,
 % ------------- %
\end{enumerate}
where $\gamma_\mathrm{E}$ stands for the Euler constant.
 \end{lemma}
 % ------------- %
 \begin{proof}
The claim can be easily verified by reformulating results of Avramidi \cite{Av1, Av2} on high-mass asymptotics of the operator $-\Delta +m^2$ as $m \to \infty$. More precisely, it follows from the stated asymptotics of equations (33) and (36) in \cite{Av2} in combination with the expression for $b_q$ given in \cite{Av1}. The constant $a_0$ in \cite{Av2} can be determined from the form of the singular parts of Green's function to fit with our convention.
 \end{proof}
 % ------------- %

 % ------------- %
 \begin{remark} %\label{}
{\rm In the case of a graph with a compact core corresponding to $d=1$, which we use for comparison, one has instead $F_1(k) = \pm \frac{1}{2ik} +\mathcal{O}(|\mathrm{Im\,}k|^{-2})$ for $\mp \mathrm{Im\,}k>0$.}
 \end{remark}
 % ------------- %

Now we can use the previous lemmata to prove the main result of this section.
 % ------------- %
 \begin{theorem}\label{thm-main2d}
Consider a manifold $\Omega,\: \dim\Omega=2,$ and let $H_U$ be the Hamiltonian on $\Omega$ with several halflines attached at a single point by coupling condition (\ref{eq-coupling}). Then all the resonances of this system are located in the $k$-plane within a finite-width strip parallel to the real axis.
 \end{theorem}
 % ------------- %
 \begin{proof}
From equation (\ref{eq-coupling2}) it follows that $\tilde u(k)$ and subsequently also the expression \mbox{$i (\tilde u(k) -1)^{-1}(\tilde u(k)+1)$} is a rational function of the momentum variable $k$. Hence there exists such a constant $C$ that for $|\mathrm{Im\,}k|>C$ the leading term of $F(k)$ behaves either like a multiple of $\ln{|\mathrm{Im\,}k|}$ or like the leading term of previously mentioned rational function. The constant $C$ can be chosen such that the contribution of the rest of $F$ does change substantially the phase of $F$. More precisely, we then have $|k^n|\ge C^n$ and
 % ------------- %
 \[
   |\ln{(\mp ik)}| =  \left|\ln|k| \mp \frac{\pi}{2} + \arg k \right| \ge \frac12 \ln C
 \]
 % ------------- %
for large enough $C$, and consequently, the dominant phase behaviour of
 % ------------- %
 \[
 a_n k^n \left( 1+ \frac{\ln(\mp ik) + \sum_{j=0}^{n-1} a_j k^j + \mathcal{O}(|k|)}{a_n k^n} \right)
 \]
 % ------------- %
and
 % ------------- %
 \[
 \ln(\mp ik) \left( 1+ \frac{c + \mathcal{O}(|k|)}{\ln(\mp ik)}\right)
 \]
 % ------------- %
is determined by the terms in front of the brackets, in particular, there are finitely many jumps of the phase of $F$ between zero and $2 \pi$ along the part of the circle $|k|=R$ in the region $|\mathrm{Im\,}k|>C$ with $C$ sufficiently large. In other words, all but finitely many resonances can be found within the strip $|\mathrm{Im\,}k|<C$, hence all resonances are located within some strip parallel to the real axis in the momentum plane.
 \end{proof}
 % ------------- %

 % ------------- %
 \begin{theorem}\label{thm-main3d}
Let $d = 3$, and let $H_U$ be Hamiltonian on $\Omega$ with several halflines attached at one point by coupling condition (\ref{eq-coupling}). Then all resonances of $H_U$ are located within a strip parallel to the real axis.
 \end{theorem}
 \begin{proof}
In the case when the coupling term does not coincide with the first term of asymptotics of $F_1$, i.e. $(\tilde u(k)-1)^{-1}(\tilde u(k)+1) \ne \pm \frac{k}{4\pi}$, one can employ the same arguments as in previous theorem. Let us check that no unitary matrix can lead to such an effective coupling matrix. Was it the case we would have
 % ------------- %
 \begin{equation}
      \tilde u(k) = -\,\frac{4 \pi \pm k}{4 \pi \mp k}\,. \label{eq-uu}
 \end{equation}
  % ------------- %
Assume that (\ref{eq-uu}) holds true for some unitary matrix $U$. For the upper sign the expression diverges at $k=4\pi$; this contradicts the unitarity of $U$ by which its modulus must not exceed one.

Let us now turn to the lower-sign case. Using $U_4 = V^{-1} D V$, $U_{2V} =  U_2 V^{-1}$ and $U_{3V} = V U_3$ with a diagonal $D$ and a unitary $V$, the equation (\ref{eq-uu}) becomes
  % ------------- %
$$
  -\frac{4\pi -k}{4\pi+k} = u_1 -U_{2V}\left(D-\frac{1+k}{1-k}I\right)^{-1} U_{3V}\,.
$$
  % ------------- %
Let us look how the right-hand side behaves in the vicinity of $-4\pi$
choosing $k = -4\pi+\varepsilon$. If none of the eigenvalues of $D$ equals $\frac{1-4\pi}{1+4\pi}$ the relation cannot be valid as $\varepsilon\to 0$. Furthermore, combining (\ref{eq-uu}) with the behaviour of the last expression as $k\to 1$, we find $u_1 = \frac{1-4\pi}{1+4\pi}$. Were there two or more eigenvalues of $D$ equal to $\frac{1-4\pi}{1+4\pi}$, we could conclude that the vector $(u_1,U_{3V})^\mathrm{T}$ has norm bigger than one, which contradicts the unitarity of $U$. Consequently, there is exactly one eigenvalue $\frac{1-4\pi}{1+4\pi}$ of the matrix $D$. Since the rows and columns of $U_4$ can be rearranged, we may assume that it is the first eigenvalue in which case we have
  % ------------- %
\begin{multline*}
  -\frac{8\pi}{\varepsilon} = u_1-1+U_{2V}^{(1)} U_{3V}^{(1)}
\frac{(1+4\pi)(1+4\pi+\varepsilon)}{2 \varepsilon}
  \\
  -U_{2V}'\left(D'-\frac{1-4\pi+\varepsilon}{1+4\pi
  -\varepsilon}I\right)^{-1}U_{3V}'\,,
\end{multline*}
  % ------------- %
where $U_{2V}^{(1)}$ and $U_{3V}^{(1)}$ are the first entries of $U_{2V}$
and $U_{3V}$, and $U_{2V}'$ and $U_{3V}'$ is rest of the row/column,
respectively. To match the $\varepsilon^{-1}$ terms on both sides of the last equation, the identity
  % ------------- %
$$
  -8\pi = \frac{(1+4\pi)^2}{2}U_{2V}^{(1)} U_{3V}^{(1)}
$$
  % ------------- %
has to be valid. From this and the unitarity of $U$ it follows that $|U_{2V}^{(1)}|= |U_{3V}^{(1)}| = \sqrt{1-\left(\frac{1-4\pi} {1+4\pi}\right)^2}$. Using the unitarity again, we note that the column
$(u_1,U_{3V})^\mathrm{T}$ must have norm equal to one, and since we know already that $u_1 = \frac{1-4\pi}{1+4\pi}$, it follows that $U_{2V}' = U_{3V}'= 0$. Equation (\ref{eq-uu}) now becomes
  % ------------- %
\begin{multline*}
  -\frac{4\pi-k}{4\pi+k} = \frac{1-4\pi}{1+4\pi} \\
  -\left(1-\left(\frac{1-4\pi}{1+4\pi}\right)^2\right)\,\mathrm{e}^{i
\varphi} \left(\frac{1-4\pi}{1+4\pi}-\frac{1+k}{1-k}\right)^{-1}
\end{multline*}
  % ------------- %
with $\varphi$ being the phase of $U_{2V}^{(1)}U_{3V}^{(1)}$. This clearly cannot be true for all $k$, which can be seen, for instance, from
observing the limit $k\to \infty$; in this way we come to a contradiction.
 \end{proof}
 % ------------- %

The two previous claims show that, unlike the case of quantum graphs, one cannot find a sequence of resonances which would escape to imaginary infinity in the momentum plane; in the next section we provide a couple of examples illustrating the comparison between the one-dimensional case and the two- and three-dimensional cases. Let us stress, however, that any such sequence tends to infinity along the real axis, and thus the above results do not answer the question stated in the opening, namely whether the resonance asymptotics has always the Weyl character.

Also, we postpone to another paper discussion of the case when halflines are attached at two and more points; we note that the Green's function on the manifold between two distinct points does not depend only on local curvature properties, but also nontrivially on the structure of the whole manifold.

\subsection{Examples}

To illustrate how the resonance asymptotical behaviour depends on the dimension of $\Omega$ we consider now to examples of a planar manifold with Dirichlet boundary conditions in dimensions one and two.

First we illustrate the difference between the dimensions one and two in a pair of examples, at a glance similar. In former case one is able to adjust the parameters to obtain a non-Weyl graph for which one half of the resonances escape to imaginary infinity, hence the number of phase jumps along the circle of increasing radius increases. On the other hand, we do not observe such a behaviour in dimension two.

 % ------------- %
 \begin{example}{\rm
In the case $\mathrm{dim}\,\Omega=1$ we consider an abscissa of length $2l$ with $M$ halflines attached in the middle -- cf.~Figure \ref{fig2}. We impose Dirichlet boundary conditions at its endpoints, $f(-l) = f(l) = 0$, and the condition (\ref{eq-coupling}) at the middle.
 % ------------- %
\begin{figure}
   \begin{center}
     \includegraphics{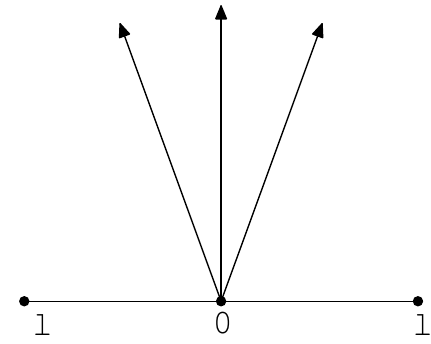}
     \caption{A `thin-hedgehog' manifold, $d = 1$}\label{fig2}
   \end{center}
\end{figure}
 % ------------- %
The Green function of the operator $H_0$ is given by
 % ------------- %
 \begin{multline*}
   G(x,y;k) =  F_1(x,y;k) =\sum_n \frac{\bar\psi_n(x)\psi_n(y)}{\lambda_n-k^2} \\
   = \frac{1}{l}\sum_{n=1}^\infty\frac{\cos{\frac{(2n-1)\pi x}{2l}}\cos{\frac{(2n-1)\pi y}{2l}}}{\left(\frac{(2n-1) \pi}{2 l}\right)^2-k^2} \,.
 \end{multline*}
 % ------------- %
Substituting, in particular, $x = y = 0$ one obtains
 % ------------- %
 \[
   F_1(0,0;k) = \frac{1}{l}\sum_{n =1}^{\infty}\frac{1}{\left(\frac{(2n-1) \pi}{2 l}\right)^2-k^2} = \frac{1}{2k}\tan{kl}\,.
 \]
 % ------------- %
Substituting this into the resonance condition one can check easily that resonance count asymptotics has a non-Weyl character if the coupling is chosen as follows,
 % ------------- %
 \[
    i\frac{\tilde u(k)+1}{\tilde u(k)-1} = \pm\frac{i}{2k} \quad \Rightarrow \quad \tilde u(k) = \frac{-2k\mp 1}{2k\mp 1}\,.
 \]
 % ------------- %
The upper-sign choice can be realized, e.g. by taking $M =2$ and connecting the halflines with the abscissa by Kirchhoff conditions -- see \cite{DEL, DP}. This corresponds to a `balanced' vertex connecting two internal and two external edges. The phase of the regularized Green's function $F_1$ and the left-hand side of the resonance condition for the above choice of the coupling, $F_1+\frac{i}{2k}$, can be seen on Figures~\ref{con1} and \ref{con2}, respectively; the latter one illustrates how the number of phase jumps increases for an increasing radius of the circle.

 % ------------- %
\begin{figure}
   \begin{center}
     \includegraphics[width=5.5cm]{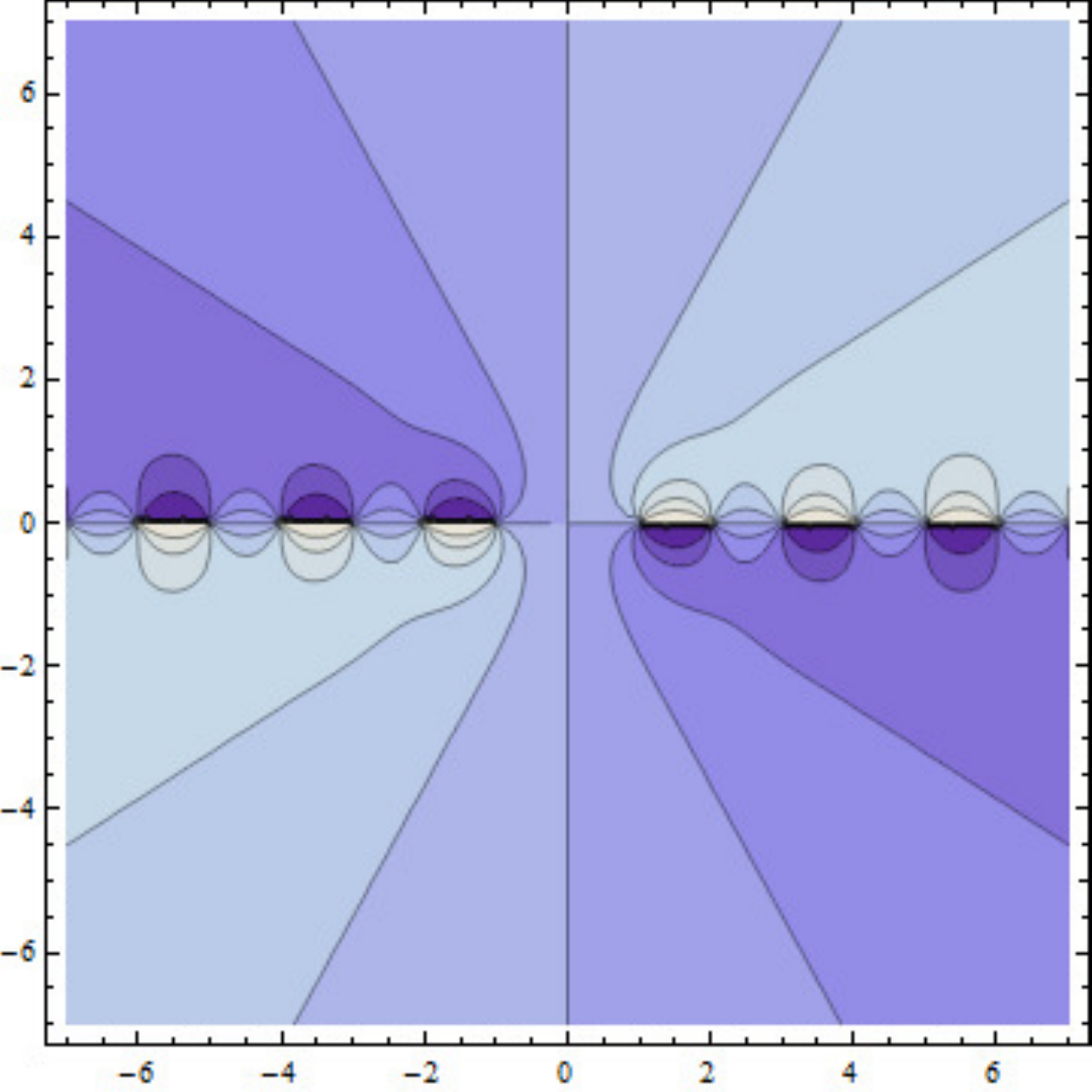}
     \caption{Phase of the Green function for $d=1$}\label{con1}
   \end{center}
\end{figure}
\begin{figure}
   \begin{center}
     \includegraphics[width=5.5cm]{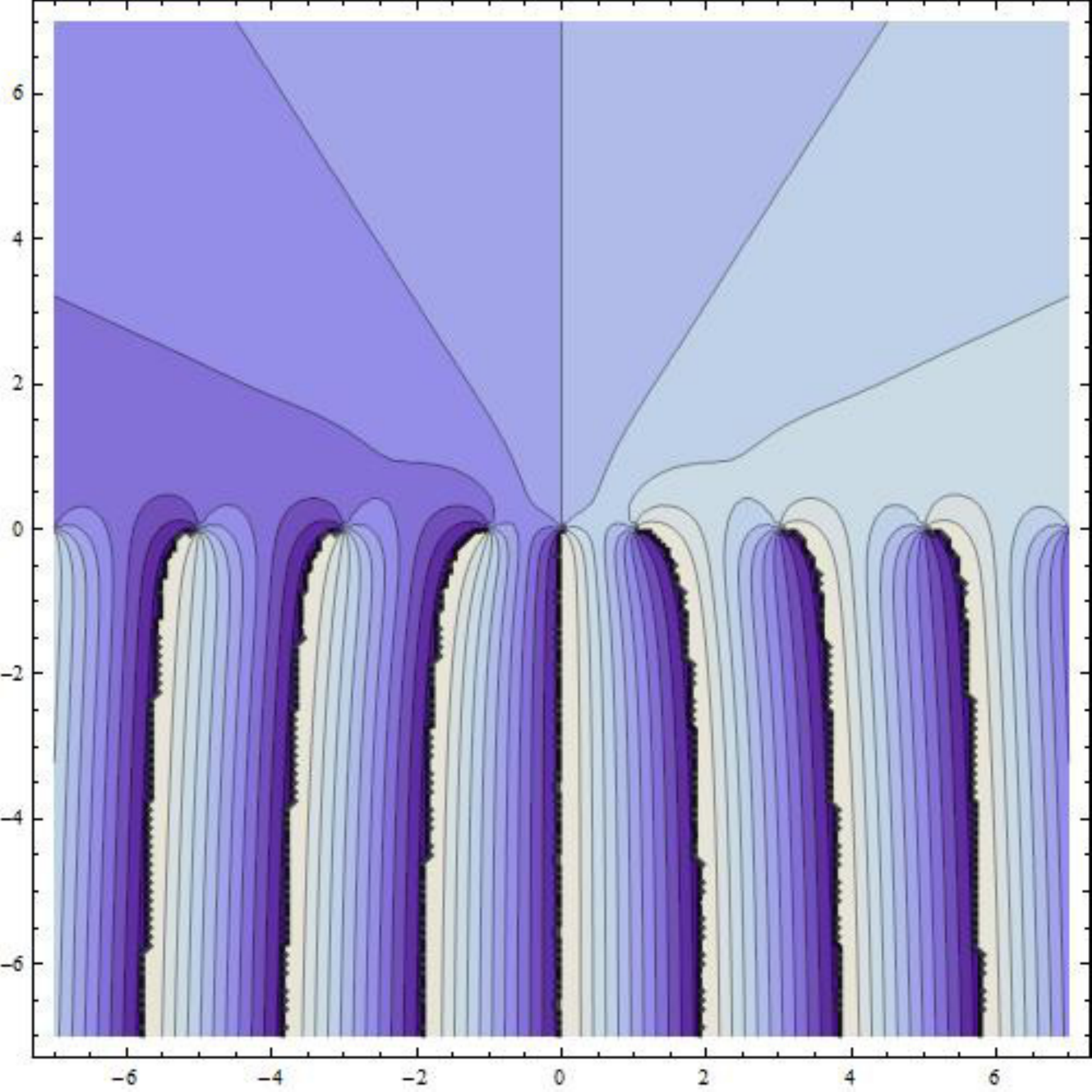}
     \caption{Phase of the Green function plus $\frac{i}{2k}$}\label{con2}
   \end{center}
\end{figure}
 }\end{example}
 % ------------- %

 \begin{example} \label{drum}
{\rm
Consider next an analogous situation in two dimensions -- a flat circular drum of radius $l$ with Dirichlet boundary condition at $r = R$ and $M$ halflines attached in its center. Because of the rotational symmetry, the
 % ------------- %
 \begin{figure}[h]
   \begin{center}
     \includegraphics{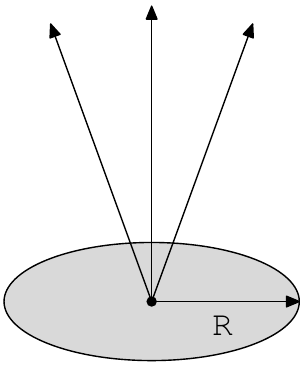}
     \caption{The hedgehog manifold of Example~\ref{drum}}\label{fig3}
   \end{center}
 \end{figure}
 % ------------- %
Green's function with one argument fixed at $y = 0$ can be expressed as a combination of Bessel functions,
 % ------------- %
 \[
   G(x,0;k) = -\frac{1}{4} Y_0(kr)+ c(k) J_0 (kr)\,,
 \]
 % ------------- %
where $r:=|x|$ and $J_0$ and $Y_0$ are Bessel functions of the first and second kind, respectively. The constant by $Y_0$ is chosen so that $G$ satisfies (\ref{eq-asym}). We employ the well-known asymptotic behaviour of Bessel functions,
 % ------------- %
 \[
   Y_0(x) \sim -\frac{2}{\pi} (\ln{x/2}+\gamma)\,,\quad  J_0(x) \sim 1
 \]
 % ------------- %
as $x \to 0$, which yields the expression
 % ------------- %
 \[
   F_1(k) = -\frac{1}{2\pi} (\ln{k}-\ln{2}+\gamma)+\frac{Y_0(kR)}{4 J_0(kR)}\,.
 \]
 % ------------- %
Using the asymptotics of $J_0(x)$ and $Y_0(x)$ as $x\to \infty$, one finds that the second term at the right-hand side behaves as $\frac{1}{4} \tan\big(kR - \frac{\pi}{4}\big)$ and its absolute value is therefore bounded for $k$ outside the real axis. The phase of $F_1$ for $R = \pi$ is plotted in Figure~\ref{con3}.
 % ------------- %
 \begin{figure}[h]
   \begin{center}
     \includegraphics[width=5.5cm]{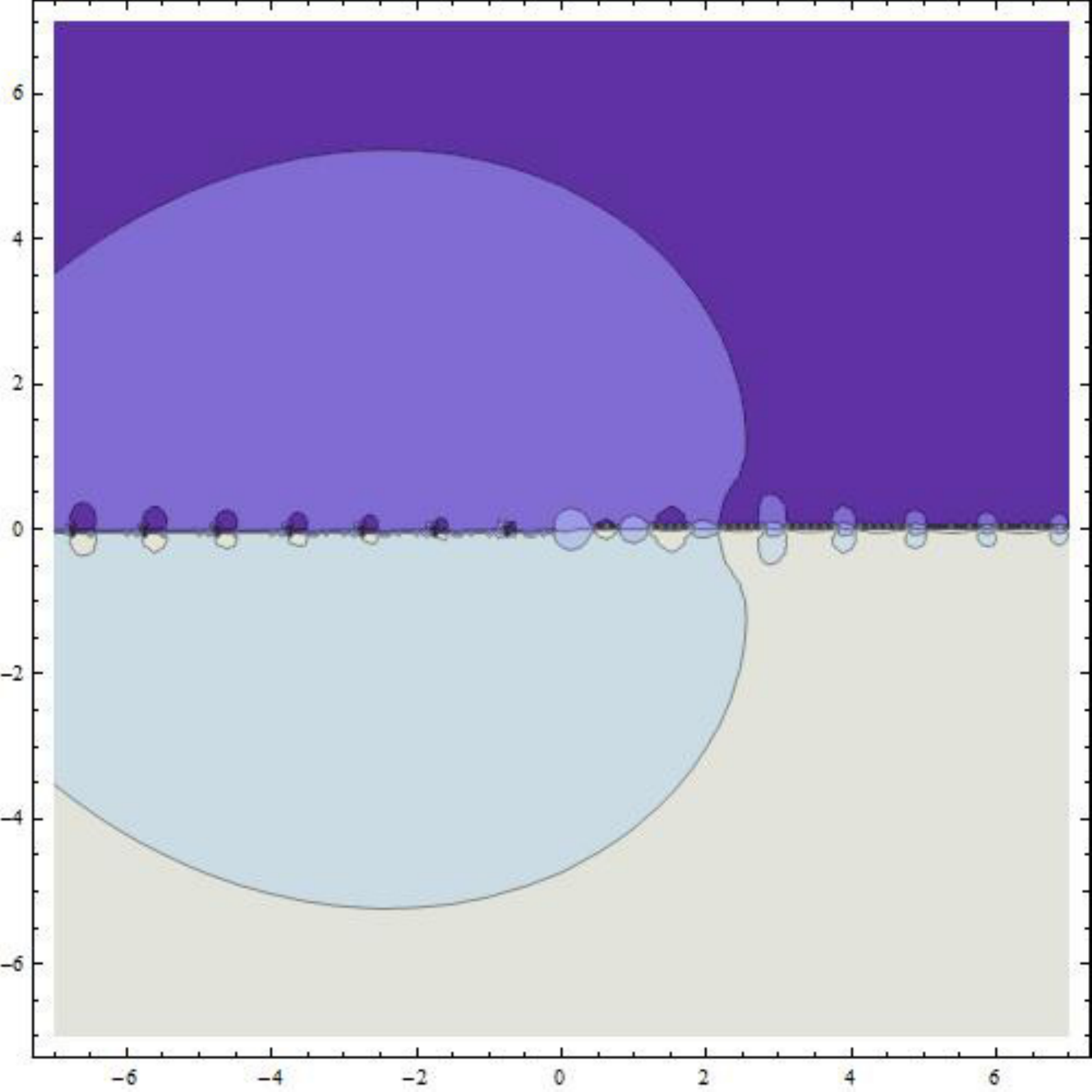}
     \caption{Phase of the regularized Green function for the hedgehog manifold of Example~\ref{drum}}\label{con3}
   \end{center}
 \end{figure}
 % ------------- %
 }\end{example}
 % ------------- %

%%%%%%%%%%%%%%%%%%%%%%%%%%%%%%%%%%%%%%%%%%%%%%%%%%%%%%%%%%
%%  RESONANCES AND MG FIELD                             %%
%%%%%%%%%%%%%%%%%%%%%%%%%%%%%%%%%%%%%%%%%%%%%%%%%%%%%%%%%%
\section{Resonances for a hedgehog manifold in magnetic field}%\addtocounter{section}{1}

Now we are going to present an example showing that an appropriately chosen magnetic field can remove all `true resonances' on a hedgehog manifold, i.e. those corresponding to poles in the open lower complex halfplane. We note that this does not influence the semiclassical asymptotics in this case because the embedded eigenvalues of the system corresponding to higher partial waves with eigenfunctions vanishing at the junction will persist being just shifted.

 % ------------- %
 \begin{figure}[h]
   \begin{center}
     \includegraphics{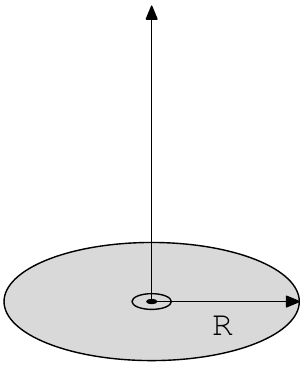}
     \caption{A disc with a lead in a magnetic field}\label{fig4}
   \end{center}
 \end{figure}
 % ------------- %

The manifold of our example will consists of a disc of radius $R$ with a halfline lead attached at its centre, for definiteness we assume that it is perpendicular to the disc plane, cf.~Figure~\ref{fig4}. The disc is parametrized by polar coordinates $r,\,\varphi$, and Dirichlet boundary conditions are imposed at $r=R$. We suppose that the system is under the influence of a magnetic field in the form of an Aharonov-Bohm string which coincides in the `upper' halfspace with the lead. The effect of an Aharonov-Bohm field piercing a surface was studied in numerous papers --- see, e.g., \cite{AT,DS,ESV} --- so we can just modify those results for our purpose. The idea is that the `true' resonances will disappear if we manage to choose such a coupling in which the radial part of the disc wave function will match the halfline wave function in a trivial way.

We write the Hilbert space of the model as $\mathcal{H} = L^2((0,R),r \mathrm{d}r )\otimes L^2(S^1)\oplus L^2(\mathbb{R}^+)$; the admissible Hamiltonians are then constructed as selfadjoint extensions of the operator $\dot{H_\alpha}$ acting as
 % ------------- %
 \[
    \dot{H_\alpha} {u \choose f} = \left( \begin{array}{c} -\frac{\partial^2 u}{\partial r^2} - \frac{1}{r} \frac{\partial u}{\partial r} + \frac{1}{r^2}\left(i \frac{\partial}{\partial \varphi}-\alpha\right)^2 u \\
     - f'' \end{array} \right)
 \]
 % ------------- %
on the domain consisting of functions ${u\choose f}$ with $u\in H^2_\mathrm{loc}(B_R(0))$ satisfying $u(0,\varphi)= u(R,\varphi) =0$ and $f \in H^2_{\mathrm{loc}}(\mathbb{R}^+)$ satisfying $f(0)=f'(0)=0$. The parameter $\alpha$ in the above expression is the magnetic flux of the Aharonov-Bohm string in the units of the flux quantum; since an integer value of the flux plays no role in view of the natural gauge invariance we may restrict our attention to the values $\alpha \in (0,1)$.

Using the partial-wave decomposition together with the standard unitary transformation $(V u)(r) = r^{1/2} u(r)$ to the reduced radial functions we get
 % ------------- %
 \[
  \dot{H_\alpha} = \bigoplus_{m = -\infty}^{\infty} V^{-1} \dot{h}_{\alpha,m} V \otimes I
 \]
 % ------------- %
where the component $\dot{h}_{\alpha,m}$ act on the upper component of $\psi={\phi\choose f}$ as
 % ------------- %
\begin{equation}
  \dot{h}_{\alpha,m} \phi = -\frac{\mathrm{d}^2 \phi}{\mathrm{d}r^2} +\frac{(m+\alpha)^2 -1/4}{r^2}\phi \,.\label{ham}
\end{equation}
 % ------------- %
To construct the self-adjoint extensions of $\dot{H_\alpha}$ which describe the coupling between the disc and the lead the following functionals can be used,
 % ------------- %
\begin{eqnarray*}
  \Phi_1^{-1} (\psi) &=&\sqrt{\pi} \lim_{r\to 0} \frac{r^{1-\alpha}}{2\pi} \int_0^{2\pi}{u(r,\varphi) \mathrm{e}^{i\varphi}}\mathrm{d}\varphi\,, \\
  \Phi_2^{-1} (\psi) &=&\sqrt{\pi} \lim_{r\to 0} \frac{r^{-1+\alpha}}{2\pi}\left[ \int_0^{2\pi}{u(r,\varphi) \mathrm{e}^{i\varphi}}\mathrm{d}\varphi \right.\\
  & & \hspace{3mm} -2\sqrt{\pi} r^{-1+\alpha}\Phi_{-1}^1(\psi)\left.\right]\,, \\
  \Phi_1^{0} (\psi) &=& \sqrt{\pi}\lim_{r\to 0} \frac{r^{\alpha}}{2\pi} \int_0^{2\pi}{u(r,\varphi)}\mathrm{d}\varphi\,, \\
  \Phi_2^{0} (\psi) &=& \sqrt{\pi}\lim_{r\to 0} \frac{r^{-\alpha}}{2\pi}\left[ \int_0^{2\pi}{u(r,\varphi)}\mathrm{d}\varphi\right.\\
  && \hspace{3mm} -2\sqrt{\pi} r^{-\alpha}\Phi_1^0(\psi)  \left.\right]\,, \\
  \Phi_1^{\mathrm{h}}(\psi) &=& f(0)\,,\quad  \Phi_2^{\mathrm{h}} (\psi) = f'(0)\,.
\end{eqnarray*}
% ------------- %
The first two of them are, in analogy with \cite{DS}, multiples of the coefficients of the two leading terms of asymptotics as $r\to 0$ of the wave functions from $\dot{H_\alpha}\!\!^*$ belonging to the subspace with $m = -1$, the second two correspond to the analogous quantities in the subspace with $m=0$, the last two are the standard boundary values for the Laplacian on a halfline.

It is obvious that if the s-wave resonances should be absent, one has to get rid of the second term in the expression (\ref{ham}) for the $m=0$ function, hence we will restrict our attention to the case $\alpha  = 1/2$. In analogy with the case of an Aharonov-Bohm flux piercing a plane treated in \cite{DS} one obtains
 % ------------- %
\begin{multline*}
  \hspace{-3mm}(\psi_1, H \psi_2) = - \int_0^{2\pi} \int_0^R  \overline{u_1}\, r^{-1/2} \frac{\mathrm{d}^2}{\mathrm{d}r^2} r^{1/2} u_2 \, r\, \mathrm{d}r\, \mathrm{d}\varphi \\
  - \int_0^\infty \overline{f_1} {f_2}'' \,\mathrm{d}x =  - \int_0^{2\pi} \int_0^R \overline{\tilde u_1} \tilde{u_2}'' \, \mathrm{d}r\, \mathrm{d}\varphi \\
  - \int_0^\infty \overline{f_1} {f_2}'' \,\mathrm{d}x =  - \int_0^{2\pi} \overline{\tilde{u_1}} \tilde{u_2}' \, \mathrm{d}\varphi \\
  +\int_0^{2\pi} \int_0^R \overline{\tilde u_1}' \tilde{u_2}'  \, \mathrm{d}r\, \mathrm{d}\varphi - \overline{f_1}(0+) f'_2(0+) \\
  +\int_0^\infty \overline{f_1}' {f_2}' \,\mathrm{d}x
 \,,
\end{multline*}
 % ------------- %
where $\tilde u_a = r^{1/2} u_a$, $a = 1,2$, is a multiple of the disc component of $u_a$ (with the prime denoting the derivative with respect to $r$) and $f_a$ is the corresponding halfline component. Hence we have
 % ------------- %
\begin{multline*}
 \hspace{-4mm}(\psi_1, H \psi_2) - (H \psi_1, \psi_2) =
 \lim_{r\to 0} \int_0^{2\pi} \left[ \overline{\tilde{u_1}} \tilde{u_2}' - \overline{\tilde{u_2}} \tilde{u_1}' \right] \, \mathrm{d}\varphi \\
 + \overline{f_2}(0+) f'_1(0+) - \overline{f_1}(0+) f'_2(0+)\,,
\end{multline*}
 % ------------- %
and using asymptotic expansion of $u$ near $r = 0$,
 % ------------- %
\begin{eqnarray*}
  \sqrt{\pi} u (r,\theta) = &(\Phi_1^{-1}(\psi) r^{-1/2} +\Phi_2^{-1}(\psi) r^{1/2}) \mathrm{e}^{-i\theta} \\
  &+\Phi_1^0(\psi) r^{-1/2} + \Phi_2^0(\psi) r^{1/2}\,,\\
 - 2 r \sqrt{\pi} u' (r,\theta)= &(\Phi_1^{-1}(\psi) r^{-1/2} - \Phi_2^{-1}(\psi) r^{1/2}) \mathrm{e}^{-i\theta} \\
 & + \Phi_1^0(\psi) r^{-1/2} - \Phi_2^0(\psi) r^{1/2}\,,
\end{eqnarray*}
 % ------------- %
one finds
 % ------------- %
\begin{multline*}
  (\psi_1, H \psi_2) - (H \psi_1, \psi_2) =\\
  = \Phi_1(\psi_1)^* \Phi_2(\psi_2) - \Phi_1(\psi_2)^* \Phi_2(\psi_1) \,,
\end{multline*}
 % ------------- %
where $\Phi_a(\psi) = (\Phi_a^{\mathrm{h}}, \Phi_a^{0}, \Phi_a^{-1})^{\mathrm{T}}$ for $a = 1,2$. Consequently, to get a self-adjoint Hamiltonian one has to impose coupling conditions similar to (\ref{eq-coupling}), namely
 % ------------- %
 \[
   (U - I) \Phi_1(\psi) + i (U + I) \Phi_2(\psi) = 0
 \]
 % ------------- %
with a unitary $U$. We choose the latter in the form
 % ------------- %
\begin{equation} \label{ABfree}
  U = \left(\begin{array}{ccc}
  0 & 1 & 0 \\
  1 & 0 & 0 \\
  0& 0& \mathrm{e}^{i\rho}
  \end{array}\right)\,,
\end{equation}
 % ------------- %
i.e. the nonradial part ($m=-1$) of the disc wave function is coupled to neither of the other two, while the radial part ($m=0$) is coupled to the halfline via Kirchhoff's (free) coupling. To see that this choice kills all the `true' resonances, we choose the Ansatz
 % ------------- %
\begin{eqnarray*}
  f(x) = a\, \sin{k x} + b\, \cos{k x}\,,\\
  u(r) = r^{-1/2}(c\, \sin{k(R-r)})
\end{eqnarray*}
 % ------------- %
which yields the boundary values
 % ------------- %
\begin{eqnarray*}
  \Phi_1(\psi) = (b, c\sqrt{\pi} \sin{kR},0)^\mathrm{T}\,,\\
  \Phi_2(\psi) =  k (a, - c \sqrt{\pi}\cos{kR},0)^\mathrm{T}\,.
\end{eqnarray*}
 % ------------- %
It follows now from the coupling conditions that
 % ------------- %
 \[
   b = c\sqrt{\pi} \sin{kR}\,, \quad a  = c\sqrt{\pi} \cos{kR}\,,
 \]
 % ------------- %
hence
 % ------------- %
 \[
  f(x) = c\sqrt{\pi} \sin{k (R+x)}\,,
 \]
 % ------------- %
thus for any $k \not \in \mathbb{R}$ and $c \ne 0$ the function $f$ contains necessarily a nontrivial part of the wave $\mathrm{e}^{-ikx}$, however, as we have argued above, a resolvent resonance can must have the asymptotics $\mathrm{e}^{ikx}$ only. In this way we come to the indicated conclusion:

 % ------------- %
\begin{proposition}
The described system has no true resonances for the coupling corresponding to the matrix (\ref{ABfree}) and the magnetic flux $\alpha=\frac12$.
\end{proposition}
 % ------------- % 

\begin{acknowledgements}
The research was supported by the Czech Science Foundation  within the project P203/11/0701 and the project Development of postdoc activities at the University of the Hradec Kr\'alov\'e, CZ.1.07/2.3.00/30.0015.
\end{acknowledgements}

\end{document}